\documentclass[a4paper,USenglish]{lipics-v2016}
\usepackage{amsmath,amssymb,amsthm}
\usepackage{microtype}
\usepackage{graphicx}
\usepackage[utf8]{inputenc}

\newcommand{\Nats}{{\mathbb{N}}}
\newcommand{\eps}{\varepsilon}
\newcommand{\argmin}{\operatornamewithlimits{argmin}}
\newcommand{\argmax}{\operatornamewithlimits{argmax}}
\newcommand{\floor}[1]{\lfloor{#1}\rfloor}
\renewcommand{\mod}{\text{ mod }}
\newcommand{\E}{{\mathbf{E}}}
\newcommand{\Var}{{\mathbf{Var}}}

\newcommand{\full}[1]{{#1}}
\newcommand{\conf}[1]{}

\def\F{{\cal F}}
\def\H{{\cal H}}
\def\P{{\cal P}}
\def\Q{{\cal Q}}
\def\S{{\cal S}}

\title{Range-efficient consistent sampling and locality-sensitive hashing for polygons\footnote{This work was supported under Australian Research Council's Discovery Projects funding scheme (project number DP150101134, Gudmundsson) and the European Research Council under the European Union’s 7th Framework Programme (FP7/2007-2013 / ERC grant agreement no.~614331, Pagh). \conf{A full version is available at arXiv~\cite{geomLSHarxiv}.}}}

\author[1]{Joachim Gudmundsson}
\affil[1]{University of Sydney, Australia\\ \texttt{joachim.gudmundsson@sydney.edu.au}}
\author[2]{Rasmus Pagh}
\affil[2]{IT University of Copenhagen, Denmark\\ \texttt{pagh@itu.dk}}

\authorrunning{Joachim Gudmundsson and Rasmus Pagh}
\Copyright{Joachim Gudmundsson and Rasmus Pagh}

\conf{
\EventEditors{Yoshio Okamoto and Takeshi Tokuyama}
\EventNoEds{2}
\EventLongTitle{28th International Symposium on Algorithms and Computation (ISAAC 2017)}
\EventShortTitle{ISAAC 2017}
\EventAcronym{ISAAC}
\EventYear{2017}
\EventDate{December 9--12, 2017}
\EventLocation{Phuket, Thailand}
\EventLogo{}
\SeriesVolume{92}
\ArticleNo{42}
}

\begin{document}

\maketitle

\abstract{
Locality-sensitive hashing (LSH) is a fundamental technique for similarity search and similarity estimation in high-dimensional spaces.
The basic idea is that similar objects should produce hash collisions with probability significantly larger than objects with low similarity.
We consider LSH for objects that can be represented as point sets in either one or two dimensions.
To make the point sets finite size we consider the subset of points on a grid.
Directly applying LSH (e.g.~min-wise hashing) to these point sets would require time proportional to the number of points.
We seek to achieve time that is much lower than direct approaches.

Technically, we introduce new primitives for \emph{range-efficient} consistent sampling (of independent interest), and show how to turn such samples into LSH values.
Another application of our technique is a data structure for quickly estimating the size of the intersection or union of a set of preprocessed polygons.
Curiously, our consistent sampling method uses transformation to a geometric problem.
}

\section{Introduction}

Suppose that you would like to search a collection of polygons for a shape resembling a particular query polygon.
Or that you have a collection of discrete probability distributions, and would like to search for a distribution that resembles a given query distribution.
A framework for addressing this kind of question is \emph{locality-sensitive hashing} (LSH), which seeks to achieve hash collisions between similar objects, while keeping the collision probability low for objects that are not very similar.
Arguably the most practically important LSH method is \emph{min-wise} hashing, which works on any type of data where similarity can be expressed in terms of \emph{Jaccard similarity} of sets, i.e., the ratio between the size of the intersection and the size of the union of the sets.
Indeed, the seminal papers of Broder et al.~introducing min-wise hashing~\cite{broder1997resemblance,broder1997syntactic} have more than 1000 citations.
Independently, Cohen~\cite{cohen1997size} developed estimation algorithms based on similar ideas (see also~\cite{cohen2007summarizing}).
The basic idea behind min-wise hashing is to map a set $S$ to $\argmin_{x\in S} h(x)$, which for a strong enough hash function $h$ gives collision probability equal (or close) to the Jaccard similarity%
\full{ (see e.g.~\cite{broder2000min} for a discussion of sufficient requirements on $h$)}%
.

If we represent discrete probability distributions by histograms there is a one-to-one relationship between the Jaccard similarity of two histograms and the statistical distance between the corresponding distributions. So a search for close distributions in terms of Jaccard similarity will translate into a search for distributions that are close in statistical distance, see Figure~\ref{fig:histogram}.

\begin{figure}[t]
  \begin{center}
  \begin{minipage}{0.8\textwidth}
  \begin{center}
  \includegraphics[width=5.5cm]{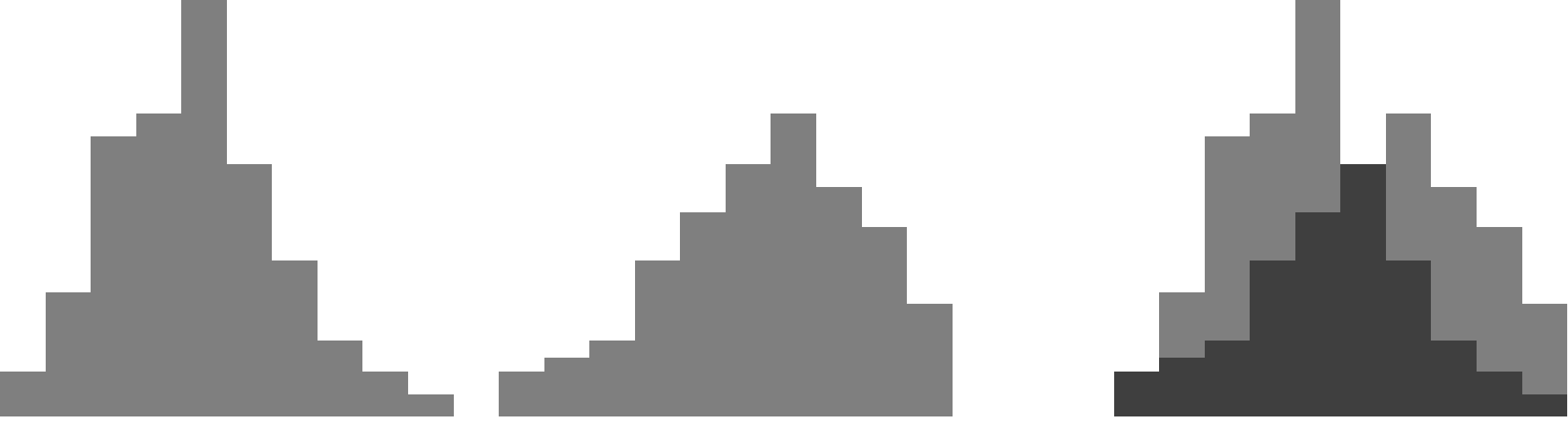}\vspace{-5mm}
  \end{center}
  \caption{Left: Two probability distributions represented as histogram point sets. Right: The statistical distance can be computed from the Jaccard similarity.}\label{fig:histogram}
  \end{minipage}
  \end{center}
\end{figure}

To make min-wise hashing well-defined on infinite point sets in the plane we may shift to an approximation by considering only those points contained in a finite grid of points.
However, for a good approximation these sets must be very large, which means that computing a hash value~$h(x)$ for each point $x\in S$, in order to do min-wise hashing, is not attractive.

\subsection{Our results}

We consider efficient locality-sensitive hashing for objects that can be represented as point sets in either one or two dimensions, and whose similarity is measured as the Jaccard similarity of these point sets.
The model of computation considered is a Word RAM with word size at least $\log p$, where $p$ is a prime number.
We use integers in $U=\{0,\dots,p-1\}$ (or equivalently elements in the field $\F_p$ of size $p$) to represent coordinates of points on the grid.
Our first result concerns histograms with $n$ values in $U$.

\begin{theorem}\label{thm:minwise-weighted}
    For every constant $\eps>0$ and every integer $N$ it is possible to choose an explicit hash function $H: U^n \rightarrow \mathbb{N}$ that has constant description size, can be evaluated in time $O(n \log p)$, and for which $\Pr[H({\bf x}) = H({\bf y})] \in [J-\eps;J+\eps]$,
    where $J = \frac{\sum_i \min(x_i,y_i)}{\sum_i \max(x_i,y_i)}$ is the weighted Jaccard similarity of vectors ${\bf x} = (x_1,\dots,x_n)$ and ${\bf y} = (y_1,\dots,y_n)$ of weight $\sum_i x_i = \sum y_i = N$.
\end{theorem}
Our construction gives an explicit alternative to existing results on weighted min-wise hashing (see~\cite{haeupler2014consistent,ioffe2010improved}) whose analysis relies on hash functions that are fully random and cannot be described in small space.
It was previously shown that a form of priority sampling based on 2-independence can be used to estimate Jaccard similarity of histograms~\cite{thorup2013bottom}, but similarity estimation is less general than locality-sensitive hashing methods such as weighted min-wise hashing.

We proceed to show the generality of our technique by presenting an LSH method for geometric objects.
We will use approximation to achieve high performance even for ``hard'' shapes, and adopt the so-called \emph{fuzzy model}
~\cite{am-ars-00}.
In a fuzzy polygon, points that are ``close'' to the boundary (relative to
the polygon’s diameter) may or may not be included in the polygon. That is, given a polygon $P$ and real value $0< \phi \leq 1$, define
the outer range $P^+=P^+(w)$
to be the locus of points whose distance from a point interior to $P$ is at most $w=\phi\cdot d(P)$,
where  $d(P)$ is the diameter of~$P$. The inner range $P^-=P^-(w)$ of $P$ is defined symmetrically.

Using the fuzzy model a \emph{valid answer} to the Jaccard similarity of two polygons $P_1$ and $P_2$ w.r.t.~$\phi$ is any value $\frac{X_{\cap}}{X_{\cup}}$ such that $A(P^-_1 \cap P^-_2) \leq X_{\cap} \leq A(P^+_1\cap P^+_2)$ and $A(P^-_1\cup P^-_2) \leq X_{\cup} \leq A(P^+_1\cup P^+_2)$, where $A(\cdot)$ denotes the area of the region.
To simplify the statement of the theorem we say that a polygon is \emph{$\alpha$-dense} in a rectangle $I$ if for some value $\alpha>0$ its area is at least a fraction $\alpha$ of the area of $I$. We use this to bound the time it takes to generate the sample points.
\begin{theorem} \label{thm:lsh-polygon}
For every choice of constants $\eps>0$, $\phi>0$ and square $I\subseteq {\bf R}^2$ it is possible to choose an explicit random hash function $H$
whose description size is constant, that can be evaluated in time $O((t\log p)/\alpha)$%
, where $t$ is the time to test if a given point lies inside a polygon, and with the following guarantee on collision probability:
Let $P_1, P_2 \subseteq I$ be polygons such that
$P_1^+$ and $P_2^+$ are $\alpha$-dense in $I$.
Then $\Pr[H(P_1)=H(P_2)] \in [J-\eps;J+\eps]$, where $J$ is some valid Jaccard similarity of $P_1$ and $P_2$ in the fuzzy model with parameter $\phi$.
\end{theorem}

It is an interesting problem whether the additive error in Theorems~\ref{thm:minwise-weighted} and~\ref{thm:lsh-polygon} can be improved to a multiplicative $1+\eps$ error.

In Section~\ref{sec:polygons} we present further applications of our technique and show how a small summary can be constructed for a set $\P$ of polygons such that for any subset $\Q$ of $\P$, an estimate of the area of $\cap \Q$ and $\cup \Q$ can be computed efficiently in the fuzzy model with respect to $\phi$.

\subparagraph*{Techniques.}
Our main technical contribution lies in methods for \emph{range-efficient min-wise hashing} in one and two dimensions,
efficiently implementing min-wise hashing for intervals and rectangles.
More specifically, we consider intervals in $U$ and rectangles in $U \times U$.
The new technique can be related to earlier methods for \emph{sampling} items with small hash values in
one or more dimensions~\cite{pavan2007range,tirthapura2012rectangle}.
(In fact, en route we obtain new hash-based sampling algorithms with improved speed, which may be of independent interest.)
However, using~\cite{pavan2007range,tirthapura2012rectangle} to sample a single item is not likely to yield a good locality-sensitive hash function.
The reason is that the hash functions used in these methods are taken from simple, 2-independent families and, as explained by Thorup~\cite{thorup2013bottom}, min-wise hashing using 2-independence does not in general yield collision probability that is close to (or even a function of) the Jaccard similarity.
Instead we use a 2-phase approach: First produce a sample of $k$ elements having the smallest hash values, and then perform standard min-wise hashing on a carefully selected \emph{subset} of the sample using a \emph{different} hash function.

We can combine and filter the samples to handle a variety of point sets that are not intervals or rectangles.
To create a sample for a subset of a rectangle we can generate a 
sample of the rectangle, and then filter away those sample points that are not in the subset.
This is efficient if the subset is suitably dense in the rectangle (which we ensure by working in the fuzzy model).
To create a sample from the union of two sets, simply take the union of the samples.
Theorems~\ref{thm:minwise-weighted} and~\ref{thm:lsh-polygon} are obtained in this way, and it would be possible to instantiate many other applications.

At the heart of our range-efficient sampling algorithms for one and two dimensions lies a reduction to the problem of finding an integer point (or integer points) in a given interval with small vertical distance to a given line. Such a point can effectively be found by traversing the integer convex hull of the line. Using a result of Charrier and Buzer~\cite{cb-arnrn-09} this can be done in logarithmic time. Thus, geometry shows up in an unexpected way in the solution.

\subsection{Comparison with related work}

We are not aware of previous work dealing with range-efficient locality-sensitive hashing.
The most closely related work is on range-efficient \emph{consistent} (or \emph{coordinated}) sampling, which is a technique for constructing summaries and sketches of large data sets.
The technique comes in two flavors: \emph{bottom-$k$} (or min-wise) sampling, which fixes the sample size, and \emph{consistent sampling} (or \emph{sub-sampling}), which fixes the sampling probability.
In both cases the idea is to choose as a sample those elements from a set $S\subseteq U$ that have small hash values under a random hash function $h: U \rightarrow [0;1]$.
If the sample size is fixed and some hash values are identical then an arbitrary tie-breaking rule can be used, e.g., selecting the minimum element.
To make $\argmin$ uniquely defined, which is convenient, we take $\argmin_{x\in I} h(x)$ to be the smallest value $y\in I$ for which $h(y) = \min_{x\in I} h(x)$.
To denote the set of the $k$ elements having the smallest hash values (with ties broken in the same way) we use the notation $\argmin_k$.
We focus on settings in which~$U$ is large and it is infeasible to store a table of all hash values.

\subparagraph*{In one dimension.}
Pavan and Tirthapura~\cite{pavan2007range} consider the 2-independent family of linear hash functions in the field of size $p$, i.e., functions of the form $h(x) = (ax+b) \mod p$. They show how to 
find hash values $h(x)$ below a given threshold $\Delta$, where $x$ is restricted to an interval $I$. (See also~\cite{Bachrach2013} for another application of this primitive.)
The algorithm of Pavan and Tirthapura uses time $O(\log p + k)$, where $k$ is the number of elements $x\in I$ with $h(x) \leq \Delta$.
Using this in connection with doubling search leads to an algorithm finding the minimum hash value in time $O(\log^2 p)$.
In this paper we show how to improve the time complexity:
\begin{lemma}\label{lemma:minwise-interval}
	Let $h(x) = (ax+b) \mod p$, where $p$ is prime and $0\leq a,b < p$.
	Given $i_2 > i_1 > 0$ consider the interval $I=\{i_1,\dots,i_2\}$.
	It is possible to compute $\argmin_{x\in I} h(x)$ (the \emph{min-hash}
	of~$I$) in time $O(\log |I|)$.
\end{lemma}
We will argue in Section~\ref{sec:SubsamplingIntervals} that Lemma~\ref{lemma:minwise-interval} can be applied repeatedly to subintervals to output the~$k$ smallest hash values (and corresponding inputs) in time $O(k \log |I|)$.
The possibility of choosing $a=0$ is included for mathematical convenience (to ensure 2-independence), though in most applications it will be better to choose $a>0$ (which in addition makes $\argmin$ uniquely defined without a tie-breaking rule).

\subparagraph*{In more than one dimension.}
Tirthapura and Woodruff~\cite{tirthapura2012rectangle} consider another class of 2-independent functions, namely linear transformations on vectors over the field $\F_2$.
Integers naturally correspond to such vectors, and for a \emph{dyadic} interval $I$ containing all integers that share a certain prefix, the problem of finding elements in $I$ that map to zero is equivalent to solving a linear system of equations.
Since an arbitrary interval can be split into a logarithmic number of dyadic intervals they are able to compute all the integers that map to zero in polylogarithmic time.
The sampling probability can be chosen as an arbitrary integer power of two.
This method generalizes to rectangles in dimension $d\geq 2$.

In this paper we instead consider linear, 2-independent hash functions of the form
$(x,y) \mapsto (ax+by+c) \mod p \enspace .$
We do not know of a method for efficiently computing a min-hash over a rectangle for such functions, but we are able to efficiently implement consistent sampling with sampling probability~$1/p$.
\begin{lemma}\label{lemma:subsampling-rectangle}
	Let $h(x,y) = (ax+by+c) \mod p$, where $p$ is prime and $0\leq a,b,c < p$.
	Given $i_1 < i_2$ and $j_1 < j_2$ consider
	$I=\{i_1,\dots,i_2\} \times \{j_1,\dots,j_2\}$. It is possible to compute
	$I' = \{(x,y)\in I \; | \; h(x,y)=0\}$ in time $O((|I'|+1) \log (\min(i_2-i_1,j_2-j_1)))$.
\end{lemma}
For random $a$, $b$, $c$ the expected size of the sample $I'$ is $|I|/p$, and because of 2-independence the distribution of $|I'|$ is concentrated around this value.
Compared to the method of~\cite{tirthapura2012rectangle} ours is faster, but has the disadvantage that the sampling probability cannot be chosen freely.
However, as we will see this restriction is not a real limitation to our applications to locality sensitive hashing and size estimation.

\subparagraph*{From consistent sampling to LSH.}
Our technique for transforming a consistent sample to an LSH value is of independent interest.
Thorup~\cite{thorup2013bottom} shows that min-wise hashing using 2-independence does not in general yield collision probability that is close to (or even a function of) the Jaccard similarity.
On the positive side he shows that bottom-$k$ \emph{samples} of two sets made using a 2-independent hash function can be used to estimate the Jaccard similarity $J$ of the sets with arbitrarily good precision.
However, this does not yield a locality-sensitive hash function with collision probability (close to) $J$, and obvious approaches such as min-wise hashing applied to the samples fails to have the right collision probability.
Instead, we use consistent sampling (using a 2-independent family) followed by a stronger hash function for which min-wise hashing has the desired collision probability up to an additive error~$\eps$.
This transformation yields the first LSH family for Jaccard similarity (with proven guarantees on collision probability) where the function can be:
\begin{itemize}
\item evaluated in $O(n+\text{poly}(1/\eps))$ time on a set of size $n$, and
\item described and computed in a constant number of machine words (independent of $n$).
\end{itemize}
Previous such functions have used either time per element that grows as $\eps$ approaches zero~\cite{MR2002c:68104}, or required description space that is a root of $n$ (see~\cite{dahlgaard2014approximately}).

\subsection{Preliminaries}
We will make extensive use of 2-independence:

\begin{definition}
A family of hash functions $\H$ mapping $U$ to $U$ is called \emph{2-independent} if $\forall x_1,x_2, a_1, a_2 \in U$ with $x_1 \neq x_2$ and $h\in\H$
chosen uniformly we have
$$\Pr[h(x_1) = a_1 \wedge h(x_2) = a_2] = 1/|U|^2 \enspace .$$
\end{definition}

It will be convenient to use the notation $x\pm \Delta$ for a number in the interval $[x-\Delta;x+\Delta]$.

Carter and Wegman~\cite{carter1977universal} showed that the family
$\H_1 = \{ x \mapsto (ax+b) \text{ mod } p \; | \; a, b \in U\}$
is 2-independent on the set $U=\{0,\dots,p-1\}$ when $p$ is a prime. Finally, we make use of $\eps$-minwise independent families:
\begin{definition}
A family of hash functions $\H$ mapping $U$ to $\Nats$ is called \emph{$\eps$-minwise independent} if for every set $S\subset U$, every $y\in S$, and random $h\in \H$:
$\Pr[h(y) = \min h(S)] = (1\pm\eps)/|S|$.
\end{definition}
Indyk~\cite{MR2002c:68104} showed that an efficient $\eps$-minwise independent family mapping to a range of size $O(p/\eps)$ can be constructed by using an $O(\log(1/\eps))$-independent family of functions (e.g.~polynomial hash functions).
Dahlgaard and Thorup~\cite{dahlgaard2014approximately} showed that the evaluation time can be made constant, independent of $\eps$, by using space $|U|^{\Omega(1)}$.
If we only care about sets of size up to some number $\hat{n}$, this space usage can be improved to $(\hat{n}/\eps)^{\Omega(1)}$.


\section{Range-efficient bottom-$k$ sampling in one dimension} \label{sec:SubsamplingIntervals}

The aim of this section is to show Lemma~\ref{lemma:minwise-interval} and how it can be used to efficiently compute consistent as well as bottom-$k$ samples.
Together with the general transformation presented in Section~\ref{sec:transform} this will lead to Theorem~\ref{thm:minwise-weighted}.

Without loss of generality suppose $0 < i_1 < i_2 < p$, consider $I=\{i_1, \ldots, i_2\} \subseteq U$, and let $h\in \H_1= \{ x \mapsto (ax+b) \text{ mod } p \; | \; a, b \in U\}$.
To show Lemma~\ref{lemma:minwise-interval} we must prove that $\argmin_{x\in I} h(x)$ can be computed in time $O(\log |I|)$.
In case $a=0$ this is trivial (just output $i_1$), so we focus on the case $a>0$.
We will show how the problem can be reduced to the problem of finding the integer point at the smallest (vertical) distance below the line segment
\begin{equation}\label{eq:ell}
 \ell = \{ (x,(ax+b)/p) \; | \; x\in [i_1;i_2] \}.
\end{equation}
\begin{figure}[t]
\begin{center}
 \includegraphics[width=0.9\textwidth]{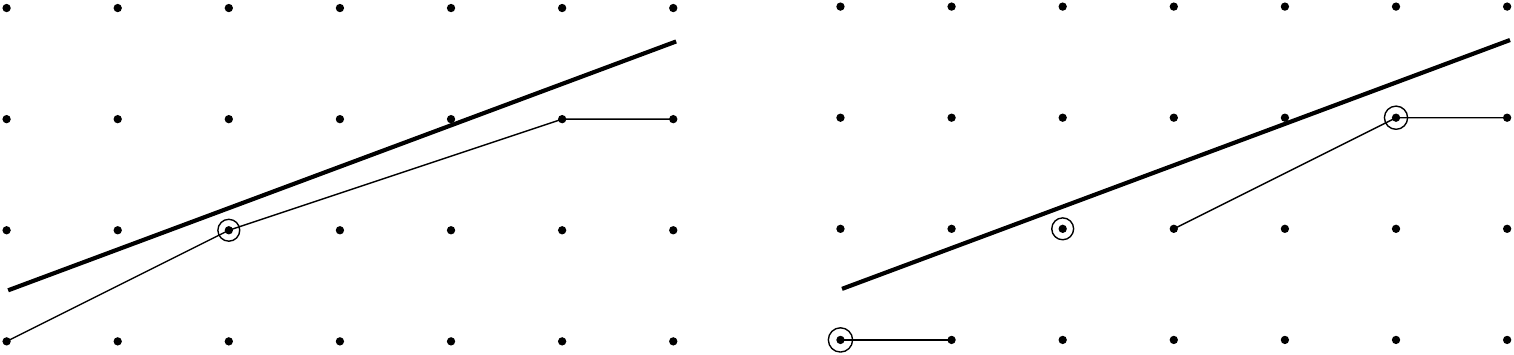}
\end{center}
\caption{Illustration of reduction to integer convex hull.}\label{fig:integerConvexHull}
\end{figure}
To see this observe that for $x\in\Nats$ we have vertical distance $(ax+b)/p - \floor{(ax+b)/p}$ between the line and the nearest integer point.
Using the equality
$$(ax+b)/p - \floor{(ax+b)/p} = ((ax+b) \text{ mod } p)/p$$
we see that minimizing $(ax+b \text{ mod } p)$ is equivalent to minimizing $(ax+b)/p - \floor{(ax+b)/p}$, as claimed.
Therefore it suffices to search for the point $(x,y)\in D = I \times \Nats$ below $\ell$ that is closest to $\ell$. 
Since $\ell$ is a line, the point $(x,y)$ must lie on the convex hull $CH(\ell)$ of the set of points  in $D$ that lie below $\ell$, referred to as the ``integer convex hull'', see Figure~\ref{fig:integerConvexHull}.
Clearly, the closest point will always be on the upper part of the hull, denoted $CH_L(\ell)$.
Zolotykh~\cite{z-nvilp-00} showed that $CH(\ell)$ consists of $O(\log (i_2-i_1))$ line segments.
To find a point on the integer convex hull with the smallest vertical distance to $\ell$ we will use a result by Charrier and Buzer~\cite{cb-arnrn-09}.

\begin{theorem} \label{thm:IntegerHull} (Charrier and Buzer~\cite{cb-arnrn-09})
  Given a line segment $\ell$, the upper integer convex hull $CH_L(\ell)$ can be computed in $O(\log(i_2-i_1))$ time, where $i_1$ and $i_2$ are the $x$-coordinates of the end points of~$\ell$.
\end{theorem}
Charrier and Buzer initially assume that $\ell$ passes through the origin.
However, they note (Section~7 in~\cite{cb-arnrn-09}) that this requirement is not needed.
Thus, using their result on the line $\ell$ defined in (\ref{eq:ell}) we obtain Lemma~\ref{lemma:minwise-interval}.

\medskip

We now discuss how to use Lemma~\ref{lemma:minwise-interval} to output the $k$ smallest hash values (and corresponding inputs, i.e., the bottom-$k$ sample) in time $O(k \log p)$.
First compute $CH_L(\ell)$ and find the point $(x_1,y_1)\in CH_L(\ell)$ with the smallest vertical distance to $\ell$. Next, split the problem into two subintervals; one for the part of $\ell$ in the $x$-interval $[i_1,x_1-1]$ and one for the part of $\ell$ in the $x$-interval $[x_1+1,i_2]$. Using a heap to find the integer point with smallest vertical distance in the intervals considered, we can repeat this process until $k$ points have been found.
To compute a consistent sample rather than the bottom-$k$ sample we simply stop the procedure whenever we see an element with a hash value larger than the threshold.

\begin{corollary} \label{cor:bottom-k}
  Let $h(x) = (ax+b) \mod p$, where $p$ is prime and $0\leq a,b < p$.
	Given $0\leq i_1 < i_2 < p$ consider $I=\{i_1,\dots,i_2\}$.
	It is possible to compute the bottom-$k$ sample (or the consistent sample of expected size $k$) from the interval $I$
	with respect to $h$ in (expected) time $O(k \log |I|)$.
\end{corollary}
It is an interesting problem whether it is possible to improve this bound to $O(k + \log |I|)$.


\section{Rectangle-efficient consistent sampling} \label{sec:SubsamplingGrids}

The aim and structure of this section are similar to those of Section~\ref{sec:SubsamplingIntervals}, but now addressing the case where we want to do hashing-based sampling in a rectangle $I=\{i_1\ldots, i_2\}\times\{j_1\ldots, j_2\}$.
Specifically, we prove Lemma~\ref{lemma:subsampling-rectangle} and show how one can use it to perform consistent sampling.
This will be used in Section~\ref{sec:transform} to prove Theorem~\ref{thm:lsh-polygon} and in Section~\ref{sec:polygons} to construct an efficient data structure for estimating the size of intersections and unions of polygons.
Assume without loss of generality that $0 \leq i_1<i_2<p$, $0 \leq j_1<j_2<p$ and $i_2-i_1\leq j_2-j_1$.
Consider the 2-independent family $\H_2=\{(x,y) \mapsto (ax+by+c) \mod p \;|\; a,b,c \in U\}$ and choose $h\in \H_2$.
To prove Lemma~\ref{lemma:subsampling-rectangle} we have to argue that
\begin{equation}\label{eq:sample}
I'=\{(x,y)\in I \;|\; h(x,y)=0\}
\end{equation}
can be computed in time $O((|I'|+1) \log (i_2-i_1))$.
Similar to the previous section we will show how the problem can be reduced to the problem of finding all integer points below a line segment $\ell$ with a small vertical distance to $\ell$.

To find all $(x,y) \in I$ for which $h(x,y) = (ax+by+c) \mod p = 0$, as a first step we translate the function $h$ such that we can consider input $y \in [0,y'_2]$.
Specifically, we replace $h$ with $h': (x,y) \mapsto (ax+by+c') \mod p$, where $c'=c + b j_1$, and consider inputs with $(x,y)\in [i_1,i_2] \times [0,j'_2]$, $j'_2=j_2-j_1$.
This is equivalent to the original task since $h(x,y) = h'(x,y-j_1)$.
Next note that for $x\in [i_1,i_2]$ and $y\in [0,j'_2]$:
$$(ax+by+c') \mod p = 0 \quad \Leftrightarrow \quad y \equiv (-b^{-1} ax -b^{-1}c') \mod p, $$
To simplify the expression set $q=-b^{-1}a$ and $s=-b^{-1}c'$.
Then we have a zero hash value when $y=(qx+s) \mod p = (qx+s)-kp$ for some positive integer $k$.
Dividing by $p$ and substituting $\alpha=q/p$ and $\beta=s/p$ we get
      $\frac {y}{p} =\alpha x + \beta-k,$
where $x\in [i_1,i_2]$ and $y/p \in [0,j'_2/p]$.
Now we can express the original problem as finding all $(x,k)\in [i_1,i_2]\times\Nats$ such that $\alpha x+\beta -k \in [0,j'_2/p]$. Consider the line segment
 $\ell' = \{ (x,\alpha x+\beta) \; | \; x\in [i_1;i_2] \}$.
An integer point $(x',y')$ below $\ell'$ with $x'\in [i_1,i_2]$ and vertical distance at most $j'_2/p$ to $\ell'$ corresponds to a point $(x',y')$ such that $h'(x',y')=0$.

To find all the points $(x',y')$ that fulfill the restrictions we can apply the same technique as in Section~\ref{sec:SubsamplingIntervals}.
That is, compute the integer convex hull $CH_L(\ell')$ using the algorithm by Charrier and Buzer~\cite{cb-arnrn-09}.
One difference from the setting of Section~\ref{sec:SubsamplingIntervals} is that we are interested in \emph{all} integer points close to $\ell'$, but $CH_L(\ell')$ is guaranteed only to include one such point if it exists.
This is handled by recursing on subintervals in which no points have been reported until we find an interval where the integer convex hull does not contain a point close to $\ell'$.
Recall that the time to output the integer convex hull is $O(\log(i_2-i_1))$ by the result of Zolotykh~\cite{z-nvilp-00}, so the cost per point reported is logarithmic.
This concludes the proof of Lemma~\ref{lemma:subsampling-rectangle}.

\subsection{Concentration bound} \label{ssec:concentration}

\begin{definition}
An ($\eps, \delta$)-estimator for a quantity $\mu$ is a randomized procedure that, given parameters $0 < \eps < 1$ and $0 < \delta < 1$, computes an estimate $X$ of $\mu$ such that $\Pr[|X-\mu| > \eps \mu] < \delta$.
\end{definition}

For some $\alpha > 0$ consider an arbitrary set $S\subseteq I$, and the sample $S' = S\cap I'$ where $I'$ is defined in (\ref{eq:sample}).
Let $1/p$ be the sampling probability. We now show that $p\, |S'|$ is concentrated around its expectation $|S|$ when $p$ is not too large.

\begin{lemma}\label{lem:concentration}
	For $\eps>0$, $p\, |S'|$ is an $(\eps, p/(\eps^2 \mu))$-estimator for $\mu=|S|$.
\end{lemma}
\begin{proof}
	The proof is a standard application of the second moment bound for 2-independent indicator variables.
	For each point $q\in S$ let $X_q$ be the indicator variable that equals $1$ if $q\in I'$ and 0 otherwise.
	Clearly we have $|S'|=X$ where $X=\sum_{q\in S} X_i$, so
	$\E[p X] = p \sum_{q\in S} \E[X_i] = \mu$.
	By definition of $I'$ the variables are 2-independent, and so $\Var(p X) = p^2 \Var[X] \leq p^2 \E[X] = p \mu$.
	Now Chebyshev's inequality implies
	$\Pr[|p X - \mu| > \eps \mu] < \Var(pX)/(\eps \mu)^2 \leq p/(\eps^2 \mu)$.
\end{proof}

To get an $(\eps, \delta)$-estimator we thus need $p \leq \delta \eps^2 \mu$.
The expected time for computing $I'$ in Lemma~\ref{lemma:subsampling-rectangle} is upper bounded by $O(\E[|I'|+1] \log p)$ which is $O((|I|/p+1) \log p)$.
If we choose $p = \Omega(\delta \eps^2 |S|)$, to get an $(\eps, \delta)$-estimator, and let $\alpha = |S|/|I|$ be the fraction of points of $I$ that are also in $S$, then the expected time simplifies to $O(\log(p)/(\alpha\delta\eps^2))$.
That is, the bound independent of the size of $S$, has logarithmic dependence on $p$, and linear dependence on $1/\alpha$, $1/\delta$, and $1/\eps$.


\section{From consistent sampling to locality-sensitive hashing}\label{sec:transform}

We now present a general transformation of methods for 2-independent consistent sampling to locality-sensitive hashing for Jaccard similarity.
Together with the consistent sampling methods in Sections~\ref{sec:SubsamplingIntervals} and~\ref{sec:SubsamplingGrids} this will yield Theorems~\ref{thm:minwise-weighted} and~\ref{thm:lsh-polygon}.

Thorup~\cite{thorup2013bottom} observed that min-wise hashing based on a 2-independent family does not 
give collision probability that is close to (or a function of) Jaccard similarity.
He observes a bias for a 2-independent family of hash functions based on multiplication, similar to the ones used in this paper.
Thus we take a different route: First produce a consistent sample using 2-independence, and then apply min-wise hashing \emph{to the sample} using a stronger hash function.
The expected time per element is constant if we make sure that the sample has expected constant size.

Let constants $\eps > 0$ and $\alpha>0$ be given.
For a point set $S \subseteq I$ with $|S|\geq \alpha |I|$ we produce a 2-independent sample $I'\cap S$ with sampling probability $1/p^*$, where $p^*=\Theta(\eps^3 \alpha |I|)$ is a prime number.
This is possible assuming $|I|>1/(\eps^3 \alpha)$ because there exists a prime $p_i$ in every interval $\{2^i,\dots,2^{i+1}-1\}$, $i=1,2,3,\ldots$.
Now select $f$ at random from an $\eps/4$-minwise independent family and define the hash value
\begin{equation}\label{eq:H2}
H^*(S) = \argmin_{x\in I'\cap S} f(x) \enspace .
\end{equation}

\begin{lemma}\label{lemma:transform}
	For $S,T \subseteq I$ with $|S|,|T|\geq \alpha |I|$ and $|I|>12\,p/(\eps^3 \alpha)$ we have $\Pr[H^*(S)=H^*(T)] = \frac{|S\cap T|}{|S\cup T|} \pm \eps$, where the probability is over the choice of $I'$ and $f$.
\end{lemma}
\begin{proof}
	Consider the Jaccard similarity of samples $S' = S\cap I'$ and $T' = T\cap I'$:
	$$J' = \frac{|S'\cap T'|}{|S'\cup T'|} = \frac{|S'|+|T'|-|S'\cup T'|}{|S'\cup T'|} \enspace .$$
	Conditioned on a fixed $I'$, the collision probability of $H^*(S)$ is $J'\pm \eps/4$ by the choice of $f$.
	Thus it suffices to show that $J'$
	differs from $J$ by at most $\eps/2$ with probability at least $1-\eps/4$

	By Lemma~\ref{lem:concentration}, $p \cdot |S'\cup T'|$ is an
	$(\eps/8, \eps/12)$-estimator for $|S\cup T|$ since $|S\cup T| \geq \alpha |I|$.
	Similarly, $p \cdot |S'|$ is an $(\eps/8, \eps/12)$-estimator for $|S|$ and $p \cdot |T'|$ is an $(\eps/8, \eps/12)$-estimator for $|T|$.
	The probability that all estimators are good is at least $1-\eps/4$, and in that case
	$$J - \varepsilon/2 < \frac{|S\cap T|-(3\eps/8) |S\cup T|}{|S\cup T|+(\eps/8)|S\cup T|} \leq J' \leq \frac{|S\cap T|+(3\eps/8) |S\cup T|}{|S\cup T|-(\eps/8)|S\cup T|} < J + \varepsilon/2 $$
	as desired.
\end{proof}

We have not specified 
$f$.
The most obvious choice is to use an $O(\log(1/\eps))$-independent hash function~\cite{MR2002c:68104}.
Another appealing choice is twisted tabulation hashing~\cite{dahlgaard2014approximately} that yields constant evaluation time, independent of $\eps$. The expected size of $S \cap I'$ is bounded by a function of $\eps$ and~$\alpha$. This means that we can combine twisted tabulation with an injective universe reduction step to reduce the domain of twisted tabulation to a (large) constant depending on $\eps$ and $\alpha$.

\subparagraph*{Proof of Theorem~\ref{thm:minwise-weighted}.}
Consider a vector ${\bf x} = (x_1,\dots,x_n) \in U^n$.
We follow the folklore approach~\cite{haeupler2014consistent} of conceptually mapping each vector ${\bf x}$ to a set $P_{\bf x}$, such that the Jaccard similarity of $P_{\bf x}$ and $P_{\bf y}$ exactly equals the weighted Jaccard similarity of ${\bf x}$ and ${\bf y}$.
In particular, it is easy to verify that this is the case if we let
$P_{\bf x} = \{ (i,j) \; | \; i=1,\dots,n ; \; j=1,\dots,x_i \}$. 
Note that $P_{\bf x}$ and $P_{\bf y}$ both have size $N$.
We will use the following class of hash functions from $U\times U$ to $U$:
\begin{equation}\label{eq:h2}
 H_2 = \{  (x,y) \mapsto (ax+by+c) \text{ mod } p \;|\; a,b,c \in U \} \enspace .
\end{equation}

The 2-independence of $H_2$ follows from the arguments of Carter and Wegman~\cite{carter1977universal}. 
\full{A proof is included in Appendix~\ref{app:GridPairwiseIndependence} for completeness.}
\conf{A proof can be found in the full version of this paper~\cite{geomLSHarxiv}.}
When restricted to points of the form $(i,\cdot)$ for a fixed~$i$, each function $h\in\H_2$ has a form suitable for Corollary~\ref{cor:bottom-k} in  Section~\ref{sec:SubsamplingIntervals}.
This means we can find the minimum for $P_{\bf x}$ restricted to a given column $i$ in time $O(\log x_i)$.
Using a heap to keep track of the smallest hash value from each column of $P_{\bf x}$ not (yet) reported in the sample, we can output all elements of $P_{\bf x}$ with a hash value smaller than any given threshold $\tau$ in time $O(\log p)$ per element.
The threshold $\tau$ is chosen to match the desired sampling probability $p^*$.

Lemma~\ref{lemma:transform} then says that we get the desired collision probability up to an additive error of~$\eps$.
The expected time to hash is $O(n \log p)$ (to populate the priority queue) plus $O(\log p)$ times the expected number of samples.
The expected number of samples $|S|/p$ is constant for every constant $\eps > 0$, which gives the desired time bound in expectation.

It is possible to turn the expected bound into a worst case bound by stopping the computation if the running time exceeds $1/\delta$ times the expectation, which happens with probability at most~$\delta$.
If we simply output a constant in this case the collision probability changes by at most~$\delta$ (which we can compensate for by decreasing $\eps$).
\hfill $\qed$

\subparagraph*{Proof of Theorem~\ref{thm:lsh-polygon}.}
The proof is similar to the proof of Theorem~\ref{thm:minwise-weighted} but with some added geometric observations.
Let $P_1$ and $P_2$ be two polygons contained in $I$.
 As mentioned in the introduction, a \emph{valid answer} to the Jaccard similarity of  polygons $P_1$ and $P_2$ with respect to $\phi$ is any value $\frac{X_{\cap}}{X_{\cup}}$ such that $A(P_1^-(w_1)\cap P_2^-(w_2)) \leq X_{\cap} \leq A(P_1^+(w_1)\cap P_2^+(w_2))$ and $A(P_1^-(w_1)\cap P_2^-(w_2)) \leq X_{\cup} \leq A(P_1^+(w_1)\cap P_2^+(w_2))$, where $w_i=\phi \cdot d(P_i)$ for $i\in \{1,2\}$.

We now switch to considering the restrictions of $P_1^+(w_1/2)$ and $P_2^+(w_1/2)$ to a $p$-by-$p$ grid of points whose enclosing rectangle contains $I$. See~\cite{h-ssr-13} for a survey on snapping points to a grid.

The grid points are identified in the natural way with integer coordinates in $[p]\times [p]$.
We choose $p$ such that the number of points inside $I$ is $\Theta(p/\alpha)$ times the desired number of samples required for Lemma~\ref{lemma:transform} to hold.

Let $L^+=[i_1,i_2]\times [j_1,j_2]$ be the minimum bounding box of $I\cap P^+_1(w_1/2)$ and $I \cap P^+_2(w_2/2)$.
The consistent sampling will be made on $P^+_i(w_i/2)$, $i\in \{1,2\}$.
The reason for this is that
$$|P_1^+(w_1/2) \cap P_2^+(w_2/2)|/|P_1^+(w_1/2) \cup P_2^+(w_2/2)|$$
is a valid answer to the Jaccard similarity of $P_1$ and $P_2$ in the fuzzy model with respect to $\phi$, which follows immediately from the below two inequalities that are proven in Lemma~\ref{lem:ApxArea} (Section~\ref{sec:polygons}):
\begin{alignat*}{1}
        &A(P_1\cup P_2) \leq |(P_1^+(w_1/2) \cup P_2^+(w_2/2)) \cap I| \leq A(P_1^+(w_1/2) \cup P_2^+(w_2/2)), \text{ and}\\
        &A(P_1\cap P_2) \leq |(P_1^+(w_1/2) \cap P_2^+(w_2/2))\cap I| \leq A(P_1^+(w_1/2) \cap P_2^+(w_2/2)) \enspace .
\end{alignat*}

Lemma~\ref{lemma:transform} gives us the desired collision probability up to an additive error of $\eps$.
The expected time to hash is $O(\log p)$ plus $O(t \log p)$ times the expected number of samples, where $t$ is the time to test if a given grid point lies inside a polygon.
If we assume that $P_1$ and $P_2$ are $\alpha$-dense in~$I$, that is, there exists an $\alpha>0$ such that  $|P^+_1(w_1/2)|,|P^+_2(w_2/2)|>\alpha \cdot |L^+|$, then the expected number of samples is $|L^+|/(\alpha p)$ for any constants $\eps$ and $\phi$, which gives the desired time bound in expectation. In many natural settings $\alpha$ is a constant, which implies that the expected number of samples is also constant.


\section{Estimating union and intersection of polygons} \label{sec:polygons}

In this section we consider the question: Given a set $\P=\{P_1, \ldots ,P_n\}$ of $n$ preprocessed polygons in the plane, how efficiently can we compute the area of the union or the intersection of a given subset $\Q\subseteq \P$?
In contrast to elementary approaches based on global, fully random sampling, our solution allows polygons to be independently preprocessed based on a small amount of shared randomness that specifies a \emph{pseudorandom} sample.

Computing the area of the union of a set of geometric objects is a well-studied problem in computational geometry. One example is the Klee’s Measure Problem (KMP). Given $n$
axis-parallel boxes in the $d$-dimensional space, the problem asks for the measure of their union. In 1977, Victor Klee~\cite{k-cmcl-77} showed that it can be
solved in $O(n \log n)$ time for $d=1$. This was generalized to $d > 1$ dimensions by Bentley~\cite{b-akrp-77} in the same year, and later improved by van Leeuwen and
Wood~\cite{lw-mprr-81}, Overmars and Yap~\cite{oy-nubkm-91} and, Chan~\cite{c-kmpme-13}. In 2010, Bringmann and Friedrich~\cite{bf-avuih-10}  gave an
$O(\frac{dn}{\epsilon^2})$ Monte Carlo $(1+\epsilon)$-approximation algorithm for the problem. 

A related question is the computation of the area of the intersection of $n$ polygons in $d$-dimensional space.
Bringmann and Friedrich~\cite{bf-avuih-10} showed that there cannot be a (deterministic or randomized) multiplicative
$(2^{d^{1-\epsilon}})$-approximation algorithm in general, unless NP$=$BPP. They therefore gave an additive $\epsilon$-approximation
for a large class of geometric bodies, with a running time of $O(\frac{nd}{\epsilon^2})$ assuming that the following three queries
can be approximately answered efficiently: point inside body, volume of body and sample point within a body.

In this section we will approach the problem slightly differently. The approach we suggest is to produce a small summary of the set $\P$,
such that given any subset $\Q$ of $\P$ the union and intersection of $\Q$ can be estimated efficiently. Unfortunately, the lower bound
arguments by Bringmann and Friedrich~\cite{bf-avuih-10} defeat any reasonable hope of achieving polynomial running time for arbitrary polygons. To get around
the lower bounds we again adopt the approximation model proposed by Arya and Mount~\cite{am-ars-00} (stated in Section~1.1)
\full{, which has been used extensively in
the literature~\cite{csx-gadbd-05,d-sasq-07,fm-strso-10}}%
.

Similar to the approach by Bringmann and Friedrich~\cite{bf-avuih-10} we will also use sampling of the polygons to estimate the size of
the union and intersection. However, compared to earlier attempts, the main advantage of our approach is that we generate the sample
points (a summary of the input) in a preprocessing step and after that we may discard the polygons. Union and intersection queries are
answered using only the summary. Also, we do not impose any restrictions on the input polygons. The drawbacks are that we only consider
the case when $d=2$ and the approximation model~\cite{am-ars-00} we use is somewhat more ``forgiving''  than previously used models.

For each polygon $P_i$ in $\P$, $1\leq i\leq n$, let $w_i=\phi\cdot d(P_i)$, where $d(P_i)$ is the diameter of $P_i$ and $0\leq \phi \leq 1$ is a given constant. Let $\Q$ be the input to a union or intersection query, that is, $\Q$ is a subset of $\P$. To simplify the notations we will write $\cup \Q^+(w) = \cup_{P_i \in \Q} P_i(w_i)$ and $\cup \Q^-(w) = \cup_{P_i \in \Q} P_i(w_i)$. Define $\cap \Q^+(w)$ and $\cap \Q^-(w)$ symmetrically.

Following the above discussion, given a legal answer to a set intersection query $X=\cap \Q$ is any~$X'$ such that
 $\cap \Q^-(w) \subseteq X' \subseteq \cap \Q^+(w)$
and for a union query $X=\cup \Q$ a legal answer is any~$X'$ such that
 $\cup \Q^-(w) \subseteq X' \subseteq \cup \Q^+(w).$
It is immediate from the above definitions that for any polygon $P$ and any $w\geq \sqrt 2$ we have:
$P^-(w) \subset P  \subset P^+(w)$.
We will use the number of integer coordinates, denoted $|P|$, within a polygon $P$ to estimate the area of the polygon, denoted $A(P)$. Proofs of Lemmas~\ref{lem:area2points},~\ref{lem:ApxArea} and~\ref{lem:IntersectionBound} can be found
\full{in the appendix.}
\conf{in the full version of this paper~\cite{geomLSHarxiv}.}

\begin{lemma} \label{lem:area2points}
 For a polygon $P$ having integer coordinates we have $A(P) \leq |P|$.
\end{lemma}
To make the queries more efficient we will not estimate the number of integer coordinates in the intersection/union $X$ of a query, instead we will estimate an approximation of $|X|$. 
We show:

\begin{lemma} \label{lem:ApxArea}
     For any polygon $P$ and $w\geq \sqrt 8$:
        $A(P) \leq A(P^+(w/2)) \leq |P^+(w/2)| \leq A(P^+(w))$.
\end{lemma}

As an immediate consequence of Lemma~\ref{lem:ApxArea} we can use the consistent samples in $P_i^+(w_i/2)$, $1 \leq i \leq n$, for our estimates of the intersection and union, provided that $w_i\geq \sqrt 8$. It remains to show how a summary of $\P$ can be computed and how the summary can be used to answer union and intersection queries.

\subparagraph*{Constructing a summary.}
For a given query $\Q$ containing $k\leq n$ polygons, let $P_{\min}=\argmin_{P_i\in \Q} |P_i^+(w_i/2)|$, $P_{\max}=\argmax_{P_i \in \Q} |P_i^+(w_i/2)|$ and let $d_{\min}=\argmin_{P_i\in \Q} d(P_i)$. If $P_i=P_{\max}$ and $P_j=P_{\min}$ then we will write $P^+_{\max}(w)=P^+_i(w_i)$ and $P^+_{\min}(w)=P^+_j(w_i)$, respectively. Before giving the construction of summary and query algorithms we state two lemmas:

\begin{lemma} \label{lem:UnionBound}
  $|P^+_{\max}(w/2)| \leq |\cup \Q^+(w/2)| \leq k \cdot |P^+_{\max}(w/2)|$.
\end{lemma}

\begin{lemma} \label{lem:IntersectionBound} If $\cap Q^-(w) \neq \emptyset$ and $\phi \cdot d_{\min}>\sqrt 8$ then
   $\frac {\phi^2}{2}\cdot |P_{\min}^+(w/2)| \leq |\cap \Q^+(w/2)|  \leq|P_{\min}^+(w/2)|$.
\end{lemma}
We will use the rectangle-efficient consistent sampling technique described in Section~\ref{sec:SubsamplingGrids} to generate a summary of $\P$  to estimate the area of $\cap \Q$ or $\cup \Q$, where $\Q$ is a given subset of~$\P$.

The idea of the construction algorithm for the summary is simple. Let $X=\cap \Q^+(w/2)$ or $X=\cup \Q^+(w/2)$ depending on the query and, assume that $\phi \cdot d_{\min}>\sqrt 8$. In a preprocessing step construct a summary of $\P$, denoted $\S$. The summary $\S$ will contain consistent samples for a number of different sampling rates. To answer a query, pick a minimum sampling rate $1/p$ that guarantees that the expected number of consistent samples in $X$ is small but sufficient to guarantee an $(\eps,\delta)$-estimate of $|X|$. If $X$ contains enough unique consistent samples then the algorithm reports an estimate of $X$, otherwise it iteratively increases the sampling rate with a constant factor until $X$ contains sufficiently many unique consistent samples. From 
Section~\ref{ssec:concentration} we know that an $(\eps,\delta)$-estimator of $X$ requires the sampling rate to be approximately $1/(\delta \cdot \eps^2 \cdot |X|)$.

From Lemmas~\ref{lem:UnionBound} and~\ref{lem:IntersectionBound} we have that the smallest area that will ever be considered in a query $\Q$ has size at least $f_{\min}=\frac {\phi^2}{2}|P^+_{\min}(w/2)|$ and the largest area is at most $f_{\max}=n \cdot |P^+_{\max}(w/2)|$. To get an $(\eps,\delta)$-estimate of $|X|$ at least $1/\delta^2\eps$ unique consistent samples are required to lie within~$X$.
As output from the above algorithm we get two data structures:
 \begin{itemize}
   \item $p[\ell]$: Returns a prime number between $[2^{\ell-1},2^{\ell}]$.
   \item $\S[P_i,\ell]$: Returns the set of consistent samples within $P^+_i(w_i/2)$, i.e., points satisfying the equation $(ax+by+c) \mod p[j]=0$. If the set is empty it returns \textsc{False}.
 \end{itemize}

\subparagraph*{Complexity.}
Consider the total number of consistent samples generated for a polygon $P_i$. The number of consistent samples is expected to increase with a factor of two in each iteration of the algorithm, that is, the expected total number of consistent samples form an exponentially growing geometric series which sums to $O(\frac{1}{\phi^2\delta^2\eps} \cdot \frac{|P_i|}{|P_{\min}|})$. Summing up over all the polygons, the total number of consistent samples is bounded by $O(\frac{n}{\phi^2\delta^2\eps} \cdot \frac{|P_{\max}|}{|P_{\min}|})$, which is also the expected size of the summary.

For the time complexity we first note that the above procedure can be implemented such that iterations where no consistent samples are expected to be generated are omitted without consideration. Since at least a fraction of $\phi/2$ of all consistent samples in the minimal bounding box of $P^+_i(w_i/2)$ is expected to lie within $P^+_i(w_i/2)$ (can be shown using a similar argument as in the proof of Lemma~\ref{lem:IntersectionBound}) the total number of generated consistent samples is expected to be at most a factor of $2/\phi$ greater than the number of consistent samples in the summary. Each consistent sample requires at most $O(\log |P_i|)$ time to generate, according to Theorem~\ref{lemma:subsampling-rectangle}. If we assume that testing if a consistent sample lies inside a polygon can be done in time $t$ then the expected time to build a summary of $P$ is $O(\frac {n}{\phi^2\delta^2\eps} \cdot \frac{|P_{\max}|}{|P_{\min}|} \cdot (t+\log |P_{\max}|) )$.
A description of union and intersection queries can be found in
\conf{the full version of this paper~\cite{geomLSHarxiv}.}
\full{the appendix.}
We can now summarize the results in this section:

\begin{theorem} \label{thm:polygon}.
Given a set $\P=\{P_1, \ldots , P_n\}$ of polygons and three constants $\eps, \delta>0$ and $0 < \phi \leq 1$. If $\phi\cdot d(P_i) \geq \sqrt{8}$ for all $P_i \in \P$ then, in the fuzzy model with respect to $\phi$, there exists a summary of size $O(\frac {n}{\phi^2\delta^2\eps} \cdot \frac{|P_{\max}|}{|P_{\min}|} \cdot (t+ \log |P_{\max}|))$ such that for any subset $\Q$ of $\P$ containing $k\leq n$ polygons an $(\eps,\delta)$-estimate of $\cup Q$ can be computed in $O(k/\delta\eps^2)$ expected time and an $(\eps,\delta)$-estimate of $\cap Q$ can be computed in $O(\frac{k}{\phi^2\delta\eps^2})$ expected time.
\end{theorem}

\full{
\section{Conclusion and open problems}

We have investigated efficient methods for consistent sampling and locality-sensitive hashing of 2-dimensional point sets.
Though the methods are simple, it is not clear if they are as useful in practice as, say, minwise hashing.
In addition to practicality, some theoretical questions remain, for example whether the additive constant $\eps$ in our theorems can be avoided.
Further, our measure of similarity among point sets is by no means the only possible one --- it would be interesting
to consider notions of similarity that are invariant under rotations and translations of the point set.
}%
\subparagraph*{Acknowledgement.} We thank the anonymous reviewers for their useful comments.

\bibliographystyle{plainurl}
\bibliography{geomLSH}

\full{
\newpage
\appendix

\section{2-independence of $\H_2$}\label{app:GridPairwiseIndependence}

\begin{lemma} \label{lem:GridPairwiseIndependence}
The family $\H_2$ defined in (\ref{eq:h2}) is 2-independent.
\end{lemma}
\begin{proof}
 Let us check that our hash function is in fact 2-independent. Let $x_1, x_2, y_1, y_2 \in U$ and $t_1, t_2 \in \F_p$ s.t. $q_1=(x_1,y_1) \neq q_2=(x_2,y_2)$.
 Let $q^{-1} \in \F_p$ be the unique multiplicative inverse of $q$.
 Note that this is guaranteed to exist if and only if $q$ is non-zero. What is the probability that $h(q_1) = t_1$ and $h(q_2) = t_2$?
 We have $(ax_1+by_1+c) \mod p =t_1$ and $(ax_2+by_2+c) \mod p =t_2$. Since $q_1\neq q_2$ we may assume without loss of generality that $x_1\neq x_2$. Now fix $b$. We get:
 \[
   \left( \begin{array}{cc}
        1 & x_1 \\
        1 & x_2
  \end{array} \right)
  \left( \begin{array}{c}
        a  \\
        c
  \end{array} \right)
  =
  \left( \begin{array}{c}
        t_1-by_1 \\
        t_2-by_2
  \end{array} \right)
\]

For every $b$ there exists exactly one pair $a, c$ such that the above equality holds. Since $a$ and $c$ are drawn uniformly and independently from $\F_p$, this probability is~$1/p^2$.
\hfill \qed \end{proof}

\section{Omitted material from Section~\ref{sec:polygons}}

\subsection{Proof of Lemma~\ref{lem:area2points}}
\begin{proof}
We use Pick's theorem~\cite{p-gzz-98}. Let $i(P)$ be the number of integer coordinates in the interior of $P$ and let $b(P)$ be the number of integer coordinates on the boundary of $P$. Pick's theorem states:
      $$A(P) = i(P)+\frac{b(P)}{2}-(h+1),$$
where $h$ is the number of holes in $P$. Since $A(P) < i(P)+ b(P) = |P|$ the lemma follows.
\end{proof}

\subsection{Proof of Lemma~\ref{lem:ApxArea}}
\begin{proof}
The first inequality is immediate and the second inequality follows from Lemma~\ref{lem:area2points}. For the third inequality we note that any point in the plane within distance $\sqrt 2$ from an integer coordinate within $P^+(w/2)$ will lie within $P^+(w)$, hence $|P^+(w/2)| \leq A(P^+(w))$.
\end{proof}

\subsection{Proof of Lemma~\ref{lem:IntersectionBound}}

\begin{proof}
Since $\cap \Q^+(w/2) \subseteq P_{\min}^+(w/2)$ the second inequality is immediate. For the first inequality we first observe that
 \begin{equation}\label{first}
   A(P^+_{\min}(w/2)) \leq  \frac{\pi}{2} ((1+\phi) \cdot d_{\min})^2.
 \end{equation}
To see this let $a$ and $b$ be two points on the boundary of $P^+_{\min}$ with largest inter-point distance among all points in $P^+_{\min}(w/2)$. The distance between $a$ and $b$ is at most $(1+\phi) \cdot d_{\min}$. Note that $P^+_{\min}(w/2)$ is enclosed in the intersection of the two disks of radius $|ab|$ centered at $a$ and $b$, hence,~(\ref{first}) follows.

Next consider a point $p$ in the non-empty set $\cap \Q^-(w)$. By definition $p$ must lie in $P_i^-(w)$ for every $P_i \in \Q$. As a result the ball of radius $\frac{3}{2}\phi \cdot d_{\min}$ and center at $p$ must lie in $\cap \Q$, thus
 \begin{equation}\label{second}
   A(\cap\Q^+(w/2)) \geq \pi (\frac{3}{2} \cdot \phi \cdot d_{\min})^2.
 \end{equation}

 Using the same argument as in the proof of Lemma~\ref{lem:ApxArea} we can show:
  \begin{equation}\label{third}
  |P^+_{\min}(w/2)|\leq A(P^+_{\min}(w/2+\sqrt{2})) \leq \frac{\pi}{2} ((1+\phi) \cdot d_{\min})^2 < \pi ((1+\phi)\cdot d_{\min})^2,
  \end{equation}
 where the second inequality follows from~(\ref{first}).

 Now we are ready to prove the first inequality of the statement of the lemma.
 \begin{eqnarray*}
   |\cap \Q^+(w/2)| & \geq & A(\cap \Q^+(w/2)) \hspace{1.8cm} \textrm{[from Lemma}~\ref{lem:ApxArea}] \\
                    & \geq & \pi (\frac{3}{2} \cdot \phi \cdot d_{\min})^2 \hspace{1.7cm} \textrm{[from}~(\ref{second})]\\
                    &  >   & \frac{\phi^2}{2}\cdot \pi((1+\phi) \cdot d_{\min})^2 \\
                    &  >   & \frac{\phi^2}{2}\cdot |P^+_{\min}(w/2)|  \hspace{1.5cm} \textrm{[from}~(\ref{third})]
 \end{eqnarray*}
This completes the proof of the lemma.
\end{proof}

\subsection{Union queries} \label{ssec:union}

A union query is a subset $\Q$ of $\P$ containing $k\leq n$ polygons.
Let $f^{\cup}_{\min}=|P^+_{\max}(w/2)|$, let  $f^{\cup}_{\max}=k \cdot |P^+_{\max}(w/2)|$ and, let $X=\cup Q^+(w/2)$. The search for a good size consistent sample starts by selecting a $j \in \mathbb{Z}$ such that $2^{j-1} \leq c \cdot \delta\eps^2 \cdot f^{\cup}_{max}  \leq 2^{j}$, for some constant $c$. The expected total number of consistent samples, $\sum_{P_i \in P} \S[P_i,p[j]]$, is $O(1/\delta\eps^2)$. This is the initial sampling rate tested by the algorithm. If the number of unique consistent samples in $X$ is at least $\frac {1}{\delta \eps^2}$ then the algorithm returns the estimate $|X|_{S} \cdot p[j]$, where $|X|_{S}$ is the number of unique consistent samples in $X$. Otherwise, the algorithm increases the sampling probability iteratively with a factor of approximately two, by decreasing $j$ by one, until the number of unique consistent samples in $X$ is at least $\frac {1}{\delta \eps^2}$. At that point the algorithm returns the estimate $|X|_{S} \cdot p[j]$.

In each iteration the number of unique consistent samples is $O(1/\delta\eps^2)$, which implies that the expected total number of consistent samples considered in one iteration is $O(k/\delta\eps^2)$. As noted above the number of consistent samples increases by roughly a factor of two in each iteration, thus the total number of consistent samples considered can be described by an exponentially growing geometric function which has an upper bound of $O(k/\delta\eps^2)$. The size of the union of consistent samples can be computed in expected linear time with respect to the total number of consistent samples, thus $O(k/\delta\eps^2)$ time in total.

\subsection{Intersection queries} \label{ssec:intersection}
An intersection query is handled similarly, but instead of fetching the consistent samples in all the polygons, we only fetch the consistent samples in $P^+_{\min}(w/2)$ since each unique consistent sample in $\cap \Q^+(w/2)$ must also be in $P^+_{\min}(w/2)$. Then, for each consistent sample $s \in \S[P^+_{\min}(w/2), \ell]$, check if $s$ is in $\S[P^+_i(w_i/2), \ell]$ for all $P_i\in \Q$. The membership query can be answered in constant expected time using hashing, thus the time required for one iteration is $k$ times the number of consistent samples in $P^+_{\min}(w/2)$.

Let $f^{\cap}_{\min}=\frac{\phi^2}{2}|P^+_{\min}(w/2)|$, let  $f^{\cap}_{\max}=|P^+_{\min}(w/2)|$ and, let $X=\cap Q^+(w/2)$. Again we start by selecting a $j \in \mathbb{Z}$ such that $2^{j-1} \leq c \cdot \gamma \delta\eps^2 \cdot f^{\cap}_{max}  \leq 2^{j}$, for some constant $c$.

The expected total number of consistent samples in $P^+_i(w/2)$ is $O(1/\delta\eps^2)$. This is the initial sampling rate tested by the algorithm. If the number of consistent samples in $X$ is at least $\frac {1}{\delta \eps^2}$ then the algorithm returns the number of consistent samples in $X$ times $p[j]$. Otherwise, the algorithm increases the sampling probability iteratively with a factor of approximately two, by decreasing $j$ by one, until the number of unique consistent samples in $X$ is at least $\frac {1}{\delta \eps^2}$.

Since $f^{\cap}_{\min}=\frac{\phi^2}{4}|P^+_{\min}(w/2)|$, the expected total number of consistent samples considered is bounded by $\frac{2}{\phi^2} \cdot \frac{1}{\delta\eps^2})$. To summarize the query time adds up to $O(\frac{k}{\phi^2\delta\eps^2}))$.

}
\end{document}